\pdfoutput=1 
\documentclass{article}

\usepackage{arxiv}

\usepackage[utf8]{inputenc} 
\usepackage[T1]{fontenc}    
\usepackage{hyperref}
\hypersetup{colorlinks,allcolors=black}

\usepackage{url}            
\usepackage{booktabs}       
\usepackage{amsfonts}       
\usepackage{nicefrac}       
\usepackage{microtype}      
\usepackage{lipsum}
\usepackage{graphicx}
\graphicspath{ {./images/} }

\usepackage{makecell}
\usepackage{listings}

\usepackage{graphicx}
\usepackage{physics}
\usepackage{amsmath}
\usepackage{breqn} 
\usepackage{xcolor} 
\usepackage{siunitx} 
\usepackage[british]{babel} 
\usepackage{tabularx} 
\usepackage{multirow}
\usepackage{algpseudocode}
\usepackage{algorithm}
\usepackage[utf8]{inputenc}
\DeclareUnicodeCharacter{2212}{-}
\usepackage{amsmath}

\usepackage{caption}
\usepackage{amsthm}
\usepackage{array}
\newcolumntype{P}[1]{>{\centering\arraybackslash}p{#1}}
\definecolor{maiblue}{rgb}{0, 0., 0.69}
\hypersetup{colorlinks=true, citecolor=maiblue, pdffitwindow=true, linkcolor=maiblue, urlcolor=maiblue}

\usepackage{hyperref}
\usepackage{color, colortbl}

\definecolor{Gray}{gray}{0.925}

\newtheorem{lemma}{Lemma}

\newcommand\algoFullName{Variational Quantum Algorithm for Coalition Structure Generation in Induced Subgraph Games}
\newcommand\algoName{QuACS }

\usepackage{authblk}
\usepackage{blindtext}

\title{\algoName: \algoFullName }

\author[]{Supreeth Mysore Venkatesh,Antonio Macaluso, Matthias Klusch} 
\affil[]{
Agents and Simulated Reality Department\\ German Research Center for Artificial Intelligence (DFKI)\\ 66123 Saarbruecken, Germany \\
\tt{\{supreeth.mysore, antonio.macaluso, matthias.klusch\}@dfki.de}
}

\begin{document}
\maketitle
\begin{abstract}
Coalition Structure Generation (CSG) is an NP-Hard problem in which agents are partitioned into mutually exclusive groups to maximize their social welfare.
In this work, we propose \algoName\hspace{-2pt}, a novel hybrid quantum-classical algorithm for Coalition Structure Generation in Induced Subgraph Games (ISGs).
Starting from a coalition structure where all the agents belong to a single coalition, \algoName recursively identifies the optimal partition into two disjoint subsets.
This problem is reformulated as a QUBO and then solved using QAOA. 
Given an $n$-agent ISG, we show that the proposed algorithm outperforms existing approximate classical solvers with a runtime of $\mathcal{O}(n^2)$ and an expected approximation ratio of $92\%$. 
Furthermore, it requires a significantly lower number of qubits and allows experiments on medium-sized problems compared to existing quantum solutions.
To show the effectiveness of \algoName we perform experiments on standard benchmark datasets using quantum simulation.
\end{abstract}



\section{Introduction}\label{sec:introduction}

A generic coalition game, also called a \textit{Characteristic Function Game (CFG)} $(A,v)$ comprises a set of intelligent agents $A$, and a characteristic function $v: \mathcal{P}(A) \to \mathbb{R}$ which maps every non-empty subset (\textit{coalition}) of $A$ to a real value.
In many practical applications, there are constraints that may limit the formation of coalitions and  the synergies between agents can be expressed as a graph\cite{deng1994complexity}.
In this case, the value of a coalition only depends on the pairwise interactions connecting the agent in the graph, and such problems are usually referred to as CFG is \textit{Induced Subgraph Game (ISG)} \cite{deng1994complexity}.

Given an ISG $(A, v)$ represented as a connected, undirected, weighted graph $G(A,w)$ where the nodes represent the agents $A=\{a_1,a_2,....a_n\}$ and the edge weights $w_{ij}$ denote the synergy between $a_i$ and $a_j$ for $i,j \in \{1, \dots , n\}$, the value of the characteristic function for a coalition $C$ can be expressed as $v(C) =  \sum_{i,j \in C} w_{ij}, \forall C\subseteq A$.
A Coalition Structure (CS) is a complete set partition of $A$ consisting of a set of coalitions $\{C_1, C_2, \dots, C_k\}$ such that ~$\bigcup_{i=1}^{k} C_{i} = A$~ and ~$C_i\cap C_j = \emptyset$~ for any ~$i,j \in \{1,2, \dots, k\}$~ and ~$i \neq j$~.
The goal is to find the optimal coalition structure $CS^*$ such that:
\begin{equation}\label{eq: CSG on ISG}
    CS^* = \arg \max_{CS} \sum_{C \in CS} \sum_{i,j \in C} w_{i,j}
\end{equation}

In the scope of this paper, we refer to the problem of CSG in ISG as the \textit{ISG problem} and we assume a fully connected graph allowing both positive and negative edge weights.
In this respect, the ISG problem remains NP-Hard \cite{bachrach2013optimal}. Notice that graph-restricted games, such as ISGs, have realistic use cases, e.g., social network analysis \cite{fortunato2010community} to discover collective groups of people with similar interests, or content downloading in self-driving cars \cite{manoochehri2017dynamic}.

\section{Related works}
ISG problems can be seen as special cases of standard CFGs and therefore, any algorithm to solve the generic CSG problem for CFGs can be applied in the context of ISGs (not vice-versa).
The time complexity of state-of-the-art exact solvers for generic CSGs \cite{rahwan2008improved,voice2012dyce,changder2021boss} scales as $\mathcal{O}(3^n)$ where $n$ is the number of agents. 

In order to deal with this exponential complexity, approximate solvers, such as  
C-link \cite{farinelli2013c} can be adopted. The C-Link algorithm is inspired by agglomerative clustering which follow a greedy bottom-up approach to determine the profit of merging (only) two coalitions at each step. This approach leads to an overall worst-case time complexity of $O(n^3)$ for a coalition game with $n$ agents at the cost of drastically reducing the exploration of the solution space. This leads to a worst-case approximation ratio of $80\%$ with respect to the quality of the solution.
Specifically for ISGs, CFSS \cite{bistaffa2014anytime} is an anytime exact solver and the state-of-the-art solver \textit{k-constrained Graph Clustering (KGC)} \cite{bistaffa2021efficient} which is an \textit{Integer Linear Programming (ILP)} based solver.
Although both CFSS and KGC perform well for sparse graphs, they might end up exploring all possible solutions for complete graphs with a complexity of $\mathcal{O}(n^n)$.

Recently, a few quantum computing solutions have been proposed for solving the CSG problem.  
BILP-Q \cite{venkatesh2022bilp} is the first general quantum algorithm for any CSG problem. It operates by reformulating the problem as a Quadratic Binary Unconstrained Optimization (QUBO) problem and solves it using both gate-based quantum computing and quantum annealing.
Although BILP-Q potentially outperforms state-of-the-art classical solutions for generic CFGs, it requires the number of logical qubits to be $\mathcal{O}(2^n)$, which is a significant limitation considering near-term quantum technology.

In addition, two possible quantum annealing solvers  for ISGs have been proposed.
GCS-Q \cite{venkatesh2022gcs} is an anytime approximate solver for any generic ISG which follows a top-down approach to find a near-optimal coalition structure. Particularly, it performs multiple calls to a D-Wave quantum annealer to find the optimal split to split for a given coalition. 
Alternatively, other existing approaches \cite{leon2017multiagent} map the graph of an ISG into specific quantum annealer hardware architecture to find the optimal coalition structure.
However, this method cannot be adopted for any problem instance as the mapping is dependent on the hardware specifications.

\section{Contribution}

In this work, we propose \algoName\hspace{-2pt} (Quantum Algorithm for Coalition Structure
generation), a novel hybrid quantum-classical algorithm for ISG problems. 
\algoName adopts the same top-down approach proposed by GCS-Q \cite{venkatesh2022gcs}, where the coalitions are  recursively split into two disjoint subsets but uses gate-based quantum optimization (QAOA \cite{farhi2014quantum}) for finding the optimal bipartition at each step. This allows investigating its scalability and usability in terms of runtime, number of gates, and number of qubits, to properly compare it with respect to existing classical and quantum solvers. Starting with a coalition structure containing a single coalition including all the agents, \algoName runs until no further split provides better coalition value for any coalition in the current optimal coalition structure. Thanks to this convenient strategy, the proposed solution is anytime and allows for obtaining near-optimal solutions in polynomial time.
These two features are essential for real-world situations where a fast near-optimal solution is needed.
We show that \algoName outperforms existing approximate solvers both in terms of runtime and quality of the solution. Furthermore, it requires a significantly lower number of qubits with respect to BILP-Q \cite{venkatesh2022bilp}.

As a second contribution, we implement \algoName using quantum simulation on standard benchmark datasets with a number of agents up to $20$. We will show that the proposed approach performs well even for shallow-depth QAOA $(p=1)$ with a worst-case approximation ratio of approximately $80\%$. Also, the performances improve when tuning $p$ up to $92\%$.

\section{Algorithm}
Given an $n$-agent ISG (cf. Definition \ref{sec:introduction}) with a fully connected underlying graph, \algoName initially assigns all the agents to a single coalition, the grand coalition $g_c$. 
Thus the algorithm considers all possible splits of $g_c$ into two disjoint coalitions, evaluating the correspondent value of the characteristic function.
If none of the generated bipartitions provides a value greater than that of $v(g_c)$, then the algorithm terminates by returning $g_c$ as the best coalition structure $CS^*$ (e.g., in the case of superadditive games).
Otherwise, the optimal bipartition $\{C, \overline{C}\}$ of $g_c$, which maximizes the characteristic function, is selected, and the optimal coalition structure $CS^*$ becomes $\{C, \overline{C}\}$, where the two sets of agents act independently from each other. 
The process of finding the optimal bipartition is then repeated  for each coalition in the current optimal coalition structure $CS^*$, and the algorithm proceeds until none of the coalitions in $CS^*$ can be split in a way to provide a better coalition value. In particular, the algorithm terminates when 
\begin{align}
    v(S)>v(C)+v(\overline{C}) \hspace{2em} \forall S \in CS^* 
\end{align}
where $C \cup \overline{C}\! = S$, $C \cap \overline{C} = \emptyset$. 
To find an optimal split (or optimal bipartition), \algoName divides the underlying connected graph of a coalition into two disconnected subgraphs by removing the edges that maximize the sum of the remaining edge weights in the subgraphs.
In other words, the nodes of the underlying graph are separated into two mutually exclusive subsets such that the sum of the edge weights in the subgraphs induced by the subset of vertices is maximum. Thus, finding the optimal split for a given ISG is equivalent to performing the weighted minimum cut (min-cut) in the underlying graph \cite{deng1994complexity}, which minimizes the sum of edge weights that are removed while partitioning the vertices into two disjoint sets.
For arbitrary edge weights, the min-cut is proven to be NP-Hard \cite{garey1979computers} and the exhaustive enumeration of all possible bipartitions for a coalition of $n$ agents is $O(2^n)$.

In order to improve the runtime for finding the best coalition structure, \algoName leverages the QAOA \cite{farhi2014quantum} for solving the min-cut problem at each step. In fact, this problem can be reformulated as a QUBO \cite{crooks2018performance} and the solution can be retrieved using hybrid quantum-classical optimization. Specifically, the solution to the min-cut provided by training the correspondent QAOA is a binary string whose values correspond to each agent belonging to one of the two partitions. Thus, this process of formulating the optimal split as a QUBO and solving it using QAOA is repeated at every step of \algoName to generate the optimal splits.

The pseudocode for \algoName is reported in Algorithm \ref{alg:cap}. 
Example \ref{example} shows how \algoName proceeds for a $4$-agent game.

\begin{algorithm}[H]
\caption{Outline of \algoName}\label{alg:cap}
\begin{algorithmic}
\Require  Set of $n$ agents $A=\{a_1,a_2,.,a_n\}$, weights ${w}:A \times A \to \mathbb{R}$
\State Initialize $CS^* \gets g_c$ \Comment{grand coalition $g_c = \{A\}$}
\For{an unexplored coalition $S \in CS^*$}
\State Derive the Ising Hamiltonian for min-cut of S
\State Solve Ising Hamiltonian using QAOA
\State Decode binary string to get $C$, $\overline{C}$ \Comment{{\small where $C\! \cup\! \overline{C}\! =\! S$, \! $C\! \cap\! \overline{C}\! =\! \emptyset$}}
\If {$v(C)+v(\overline{C}) \geq v(S)$}
\State Remove $S$
\State Add $C,\overline{C}$ to $CS^*$
\EndIf
\EndFor
\Ensure Optimal Coalition Structure $CS^*$
\end{algorithmic}
\end{algorithm}

\paragraph{Example}\label{example}
 Given an ISG with four agents $\{a_1,a_2,a_3,a_4\}$ and the edge weights between them defined as\footnote{We consider a simple example with all agents/nodes having no self-loop, i.e., their utility is equal to zero when working separately.}:
\begin{align*}
    w_{12}\!=\!2,\hspace*{3pt} w_{13}\!=\!6,\hspace*{3pt} w_{14}\! =\! -4,\hspace*{2pt} w_{23}\! =\! -5,\hspace*{2pt} w_{24}\! =\! -1,\hspace*{2pt} w_{34}\! =\! 1
\end{align*}

The \algoName finds the optimal coalition structure $CS^*$ as follows.
At each step, the split with the highest value is chosen. 
$CS^*$ always contains the best solution found until then and its value is given by $v(CS^*)$.
The algorithm terminates when none of the coalitions in $CS^*$ has a better split.
Here, $\{\{a_1,a_3\},\{a_2\},\{a_4\}\}$ is the best way to partition the agents in $A$.

\begin{table}[ht]
\centering
\begin{tabular}{|c|c|c|c|c|}
\hline 
\textbf{Step} & \textbf{Coalition (S)} & \textbf{Values to Compare} & \textbf{$CS^*$} & \textbf{$v(CS^*)$} \\ [0.5ex]
\hline
\hline
0 & - & - & $\{\{a_1,a_2,a_3,a_4\}\}$ & $-1$ \\
\hline
1 & $\{a_1,a_2,a_3,a_4\}$ & \makecell{$v(\{a_1,a_2,a_3,a_4\}) = -1$\\$v(\{a_1\})+v(\{a_2,a_3,a_4\}) = -5$\\$v(\{a_1,a_2\})+v(\{a_3,a_4\}) = 3$\\$v(\{a_2\})+v(\{a_1,a_3,a_4\}) = 3$\\ $\colorlet{oldcolor}{.} \color{green}\boxed{ \color{oldcolor} v(\{a_1,a_3\})+v(\{a_2,a_4\}) = 5}$\\$v(\{a_3\})+v(\{a_1,a_2,a_4\}) = -3$\\$v(\{a_1,a_4\})+v(\{a_2,a_3\}) = -9$\\$v(\{a_4\})+v(\{a_1,a_2,a_3\}) = 3$} & $\{\{a_1,a_3\},\{a_2,a_4\}\}$ & $5$ \\
\hline
2 & $\{a_1,a_3\}$ & \makecell{$\colorlet{oldcolor}{.} \color{green}\boxed{ \color{oldcolor} v(\{a_1,a_3\}) = 6}$\\$v(\{a_1\})+v(\{a_3\}) = 0$} & $\{\{a_2.a_4\},\{a_1,a_3\}\}$ & $5$ \\
\hline
3 & $\{a_2,a_4\}$ & \makecell{$v(\{a_2,a_4\}) = -1$\\$\colorlet{oldcolor}{.} \color{green}\boxed{ \color{oldcolor}v(\{a_2\})+v(\{a_4\}) = 0}$} & $\color{blue}\{\{a_1,a_3\},\{a_2\},\{a_4\}\}$ & $6$ \\
\hline
\end{tabular}
\vspace{1em}
\caption{The table illustrates the working of \algoName for an ISG with four agents. The green box highlights the splits of $S$ with the maximum value chosen at each step.}
\label{tab: example}
\vspace{-8mm}
\end{table}

\subsection{Performance Analysis}\label{sec:performance}
In this section, we analyze the performance of \algoName in terms of the number of qubits, number of gates, and runtime.

\begin{lemma}\label{lemma:qubit complexity}
\textit{\algoName that uses a p-layered QAOA circuit, requires $\mathcal{O}(n)$ qubits to solve an $n$-agent ISG problem.}
\vspace{-3pt}
\end{lemma}
\begin{proof}
In the execution of \algoName\hspace{-2pt}, the task of finding the optimal split is reduced to the min-cut problem, reformulated as QUBO, and solved using QAOA. The number of qubits required to solve a QUBO matrix of size $n \times n$ is equal to $n$. For a given $n$-agent problem, the largest QUBO to be solved is the one corresponding to the first step of \algoName where the grand coalition has to be split. 
Given the top-down approach of \algoName\hspace{-2pt}, any further execution of the QAOA operates on coalitions whose size is lower than $n$, which means that the number of qubits required is strictly lower than $n$. Thus, the qubit complexity of \algoName is $\mathcal{O}(n)$.
\end{proof}

\begin{lemma}\label{lemma:gate complexity}
\textit{\algoName that uses a p-layered QAOA circuit, requires $\mathcal{O}(n^2p)$ single and/or two qubit gates to solve an $n$-agent ISG problem.}
\vspace{-3pt}
\end{lemma}

\begin{proof}
According to the Lemma \ref{lemma:qubit complexity}, the largest instance of QAOA in \algoName requires $n$ qubits. In this case, the first step of QAOA generates an equal superposition of $2^{n}$ possible states through the use of $n$ Hadamard gates. Then, for each non-zero interaction in the QUBO matrix of the min-cut (cost Hamiltonian $H_c$), three gates (two CNOT gates and a local single-qubit $R_Z$ gate) are used, plus an additional $R_Z$ applied to each qubit. Notice that, in the case of a fully-connected graph, all the off-diagonal elements of the QUBO matrix are diverse from zero.
The number of the off-diagonal elements of the QUBO is $n(n-1)/2$, which can be approximated as $n^2$. 
Finally, $n$ Pauli-$X$ single-qubit rotation gates $R_X$ are applied (mixing Hamiltonian $H_B$). For a $p$-layered QAOA\footnote{for more details on the QAOA implementation see \cite{crooks2018performance}}, the whole set of gates is repeated $p$ times (except for the Hadamard)
Thus, the total number of single or two-qubit of \algoName is $n + p(n^2+n)$, which can be approximated as $\mathcal{O}(n^2p)$.
\end{proof}

In terms of runtime, the best case for \algoName is when none of the splits of the $g_c$ (grand coalition) provide a higher coalition value than $g_c$.
For a generic $n$-agent ISG, solving the min-cut is NP-Hard\cite{garey1979computers}, i.e.,  it requires at most $\mathcal{O}(2^n)$ operations to evaluate all possible bipartitions of $g_c$.
When the optimal coalition structure is defined by the set of singletons, \algoName has to process $n\!-\!1$ times the optimal split, from the grand coalition to the set of singletons.
Classically, the overall runtime of \algoName when executed using classical computation is:
\begin{equation}\label{eqn:complexity classical}
    \begin{split}
        \sum_{k=2}^n \mathcal{O}(2^{k}) = \mathcal{O}(n2^n).
    \end{split}
\end{equation}

Instead, using a $p$-layered QAOA for solving the min-cut provides a runtime of $\mathcal{O}(np)$\cite{crooks2018performance}, assuming a negligible cost for the optimization process.
Therefore, assuming $p=1$, \algoName implements $\mathcal{O}(n)$ times the QAOA with a runtime of $\mathcal{O}(n)$, which leads to an overall worst-case runtime of \algoName is $\mathcal{O}(n^2)$.

\subsection{Evaluation}
For the experiments, we generate ISGs with edge weights sampled from Normal and Random (Uniform) distributions centered in $0$ for generating both positive and negative values.
To assess the quality of \algoName\hspace{-2pt}, we implement IDP \cite{rahwan2008improved} to find the optimal coalition structure and two variants of \algoName differing from each other for the method adopted for finding the optimal split at each step: $\algoName_{c}$ solves the min-cut using the classical QUBO solver while $\algoName_{q}$ uses the QAOA.
All the implementation use IBM Qiskit an \lstinline{aer_simulator} to perform the experiments.
The runtime for both distributions with ISGs up to $20$ agents, considering a $1$-layer QAOA for $\algoName_{q}$  are reported Figure \ref{fig: IDP vs \algoName classical vs qaoa}.

\vspace{-8pt}
\begin{figure}[H]
\centering
\includegraphics[width=0.7\textwidth]{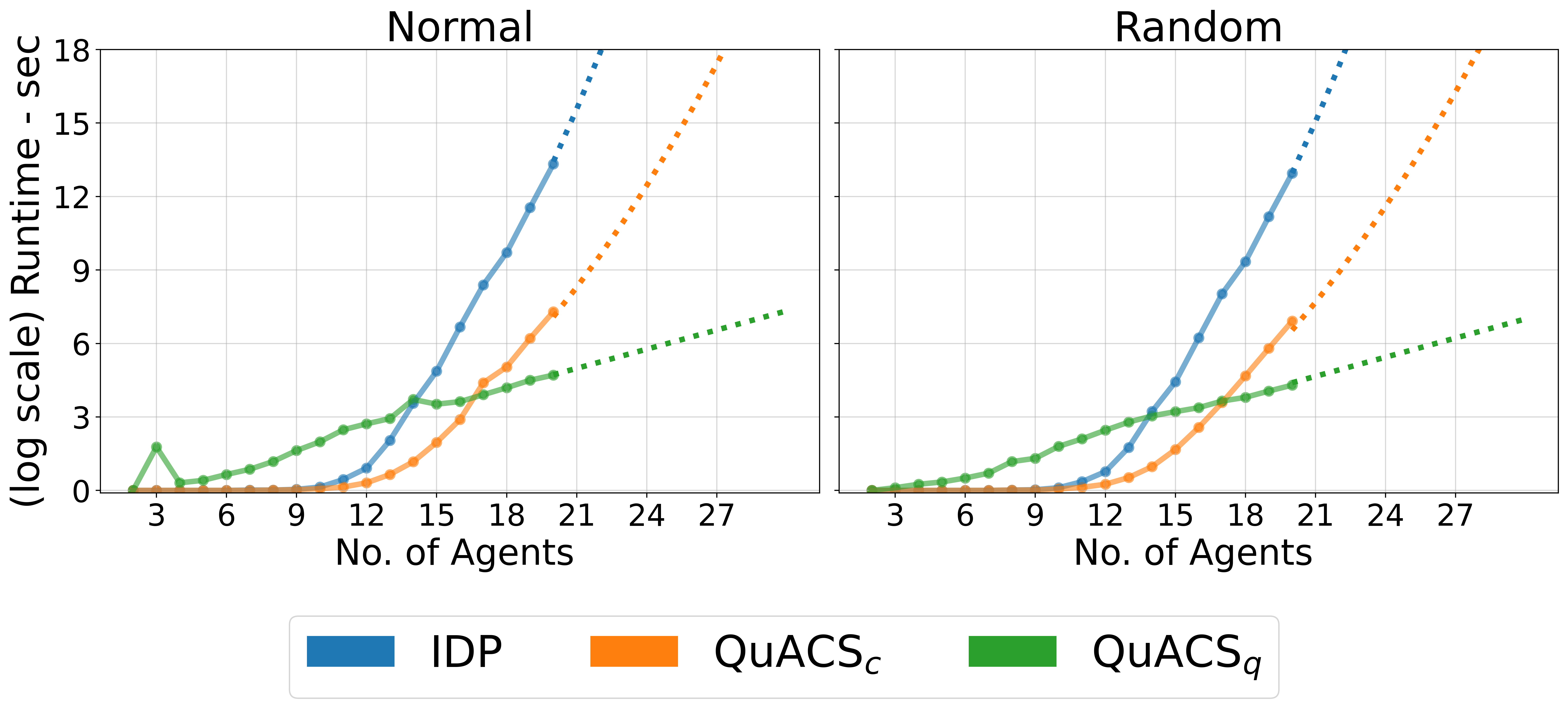}
\caption{Empirical runtimes of IDP, $\algoName_{c}$ (Eq. \ref{eqn:complexity classical}) and $\algoName_{q}$. The plots are analytically extrapolated denoted by dotted lines.}
\label{fig: IDP vs \algoName classical vs qaoa}
\end{figure}

It can be observed that the runtimes of IDP and $\algoName_{c}$ scale exponentially as expected, while $\algoName_{q}$ is polynomial to $n$, which complies with the runtimes discussed in Section \ref{sec:performance}.

Since \algoName adopts a greedy top-down strategy, it is an approximate solver, and the quality of the solutions needs to be analyzed.
Theoretical quality analysis of an approximate solver for CSG problem is possible when the game is restricted to allow only positive coalition values \cite{sandholm1999coalition}.
As we consider a more generic problem instance, we can only investigate empirically.
The value of the optimal coalition structure from IDP ($v(CS_{E})$) is used as the baseline to evaluate the approximation error in the value obtained from an approximate solver ($v(CS_{A})$) using the relation $ Er = \frac{|v(CS_{E}) - v(CS_{A})|}{v(CS_{E})} \in [0,1]$.
With IDP as an exact solver, relative errors of $\algoName_{c}$ and $\algoName_{q}$, named $Er_c$ and $Er_q$ respectively, are calculated. 
We compute the quality solution using two different approaches for $\algoName_{q}$: assuming $p=1$ and training the QAOA using different values for $p$, i.e., $p \in [1,5]$. 
Results are shown in Figure \ref{fig: quality of \algoName classical vs qaoa}. 
\vspace{-8pt}
\begin{figure}[H]
\centering
\includegraphics[width=0.7\textwidth]{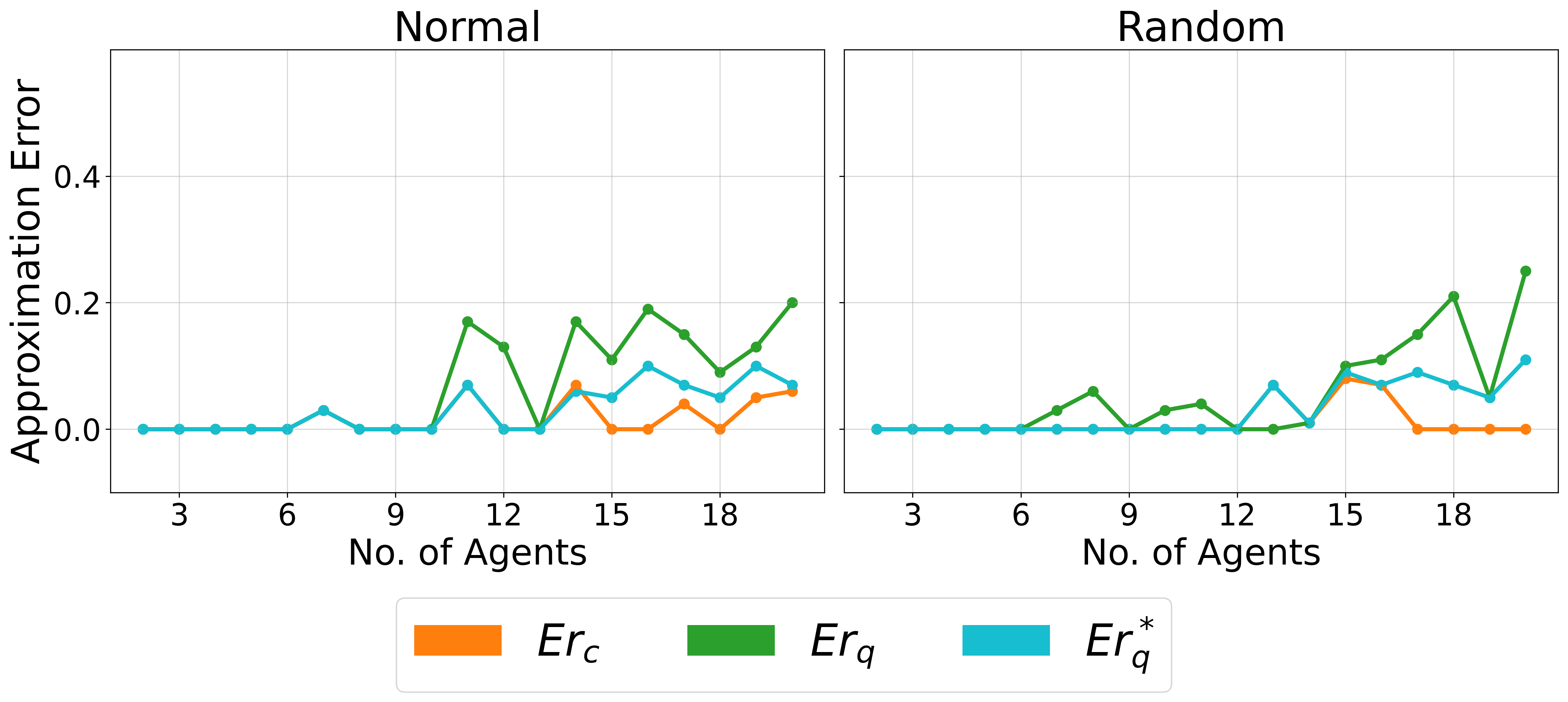}
\caption{Approximation errors for $\algoName_{c}$ and $\algoName_{q}$. $Er_q^*$ indicates the results of $\algoName_{q}$ when tuning $p$ up to $5$.}
\label{fig: quality of \algoName classical vs qaoa}
\end{figure}

The minimum value for $Er_c$ is equal to $8\%$ which translates to an approximation ratio for \algoName of $92\%$ when solving the min-cut problem classically. \algoName is more prone to error due to the intrinsically probabilistic nature of quantum simulators and the need to train the QAOA properly. In this case, for $p=1$, the maximum $Er_q$ is $7\%$ up to $10$ agents and the quality deteriorates as long the size of the problem increases.
When tuning $p$, the quality of the solution improves, and $Er_q^*$ seems to converge to $Er_c$.

\section{Discussion}

In the previous section, we showed that \algoName solves the coalition structure generation for ISGs in polynomial time with an expected worst-case approximation ratio of $92\%$.
\algoName always provides a valid coalition structure (complete set partition of the agents) at any step of the computation since the split operation produces a disjoint partition of the agents. Thus \algoName is an anytime approach.
When compared with alternative gate-based quantum solutions (e.g., BILP-Q\cite{venkatesh2022bilp}), which require $\mathcal{O}(2^n)$ logical qubits to be implemented, \algoName can be implemented using at most $n$ qubits for an $n$-agent coalition game. However, unlike BILP-Q, \algoName is an approximate solver, suitable for ISGs only.
Classically, the best CSG solver for CFG is BOSS \cite{changder2021boss}, an exact and not anytime solver with time complexity of $\mathcal{O}(3^n)$. The approximate solver C-Link \cite{farinelli2013c} has a time complexity of $n^3$, with an experimental approximation of $80\%$. 

Thus, \algoName\hspace{-2pt}, with a runtime that scales quadratically in the number of agents and an experimental approximation ratio of $92\%$ outperforms existing classical solvers and is suitable for near-term quantum technology.

\section{Conclusion}

In this work, we presented \algoName\hspace{-2pt}, an approximate and anytime quantum solver for finding the optimal coalition structure in ISGs.
Starting from the set of all agents, \algoName tries to find optimal bipartition iteratively while delegating the optimal split problem (as min-cut) to the QAOA.
\algoName scales quadratically in the number of agents $n$ and outperforms existing approximate solvers in terms of both runtime and approximation ratio.
We also implemented \algoName using quantum simulation, showing that the proposed algorithm is already a credible alternative solution for problems with tens of agents. 
As future work, techniques like warm-starting the QAOA \cite{egger2021warm} can be adopted for better training the quantum circuits.
An additional improvement can be given by parallelizing the task of finding the optimal splits on multiple quantum computers/simulators.
Finally, the adoption of \algoName can be used for generic CFGs where it is possible to find the corresponding approximate equivalent ISG \cite{bistaffa2021efficient}.

\section*{Code Availability}

All code to generate the data, figures, analyses, as well as, additional technical details on the experiments are publicly available at \href{https://github.com/supreethmv/QuACS}{\textcolor{blue}{https://github.com/supreethmv/QuACS}}.

\section*{Acknowledgments}

This work has been funded by the German Ministry for Education and Research (BMB+F) in the project QAI2-QAICO under grant 13N15586.

\bibliographystyle{unsrt}  
\bibliography{references}  






\end{document}